\documentclass[nofootinbib,aps,prb,twocolumn,amsmath,amssymb,floatfix,superscriptaddress,]{revtex4-1}

\usepackage[utf8]{inputenc}
\usepackage[dvipsnames]{xcolor}
\usepackage{graphicx}
\usepackage{amsmath}
\usepackage{amsthm}
\usepackage{dsfont}
\usepackage{bbold}
\usepackage{braket}
\usepackage{bm}
\usepackage{hyperref}

\newtheoremstyle{mystyle}%
{3pt}
{3pt}
{\upshape}
{}
{\bfseries}
{.}
{.5em}
{}

\usepackage[all]{hypcap}
\usepackage[capitalise,compress]{cleveref}
\crefname{section}{Section}{Section} 
\crefname{subsection}{Section}{Section}
\theoremstyle{mystyle}
\crefname{thm}{Theorem}{Theorems}
\Crefname{thm}{Theorem}{Theorems}
\newtheorem{thm}{Theorem}

\crefname{corll}{Corollary}{Corollaries}
\theoremstyle{remark}

\theoremstyle{definition}

\theoremstyle{mystyle}

\DeclareMathOperator{\trace}{tr}

\DeclareMathOperator{\tr}{tr}

\newcommand{\new}[1]{{\color{black}#1}}

\begin{document}

\title{\new{Error Mitigation Thresholds in Noisy Random Quantum Circuits}}
\author{Pradeep Niroula}
\affiliation{Joint Center for Quantum Information and Computer Science, NIST/University of Maryland, College Park, Maryland 20742, USA}
\author{Sarang Gopalakrishnan}
\affiliation{Department of Electrical Engineering, Princeton University, Princeton, NJ 08544, USA.}
\author{Michael J. Gullans }
\affiliation{Joint Center for Quantum Information and Computer Science, NIST/University of Maryland, College Park, Maryland 20742, USA}
\date{\today}

\begin{abstract}
Extracting useful information from noisy near-term quantum simulations requires error mitigation strategies. A broad class of these strategies rely on precise characterization of the noise source. \new{We study the robustness of probabilistic error cancellation and tensor network error mitigation when the noise is imperfectly characterized}. We adapt an Imry-Ma argument to predict the existence of a threshold in the robustness of these error mitigation methods for random spatially local circuits in spatial dimensions $D \geq 2$: noise characterization disorder below the threshold rate allows for error mitigation up to times that scale with the number of qubits. For one-dimensional circuits, by contrast, mitigation fails at an $\mathcal{O}(1)$ time for any imperfection in the characterization of disorder.  As a result, error mitigation is only a practical method for sufficiently well-characterized noise.    We discuss further implications  for tests of quantum computational advantage, fault-tolerant probes of measurement-induced phase transitions, and quantum algorithms  in near-term devices.
\end{abstract}

\maketitle

\section{Introduction}

Noise presents a fundamental barrier to realizing scalable quantum information processing \cite{NielsenChuang}.
The theory of fault-tolerant quantum error correction shows how this barrier can be overcome in principle \cite{Aharonov1997,Dennis02,Gottesman09}.   However, despite remarkable experimental progress establishing the basic validity of the theory of fault-tolerance \cite{egan2021faulttolerant,Krinner22,Acharya22,Sivak22}, achieving large-scale quantum computing with error corrected qubits remains a formidable challenge.  In recent years, quantum error mitigation has emerged as a complementary paradigm for addressing the effects of noise in large-scale quantum devices \cite{McArdle2020,cai2022quantum,
Temme17,Li17,Kim23,McClean17,Bonet18,McArdle19,Cotler19,Huggins21,Koczor21,Czarnik21,Srikis21}.  At its core, error mitigation relies on the ability to accurately characterize the interaction of the system with its environment.  Armed with this knowledge, one can design classical post-processing techniques to construct more accurate estimators of the noiseless signal using data obtained from a noisy device.  \new{For example, similar methods have been ubiquitously employed in quantum process tomography to reliably extract error models in the presence of noisy operations \cite{Emerson05,Knill08,Nielsen21}.}

A key aspect of the threshold theorem for fault-tolerant quantum computing is that it applies even in the case where the individual components used to implement the error correction are noisy.
An analogous question for error mitigation strategies is whether they work when the noise is imperfectly characterized.
As noted above, related questions have been addressed for quantum process tomography in the presence of faulty operations \cite{Emerson05,Knill08,Nielsen21}; however, these models are often considered only at the level of a few qubits,  generally precluding the existence of a sharply defined phase transition. In the case of many-body tomographic problems (e.g., Pauli noise estimation \cite{Flammia2019,Harper2020} or Hamiltonian learning \cite{Bairey19,Anshu21}), thresholds in learnability are possible and likely arise under some circumstances.  However, to our knowledge, threshold results for the robustness of quantum error mitigation have not been previously considered.

In this paper, we demonstrate the possibility of threshold in the robustness of quantum error mitigation.
\new{We focus on two conceptually simple forms of quantum error mitigation called ``probabilistic error cancellation'' (PEC) \cite{Temme17,Berg22} and tensor-network error mitigation (TEM) \cite{Filippov2023}.} 
PEC and TEM rely on the fact that many common noise channels are  invertible.    \new{Unfortunately, the inverse operation is not easily implemented by  physical evolution without access to the environment.  Instead the inverse channel is achieved through classical post-processing.  This post-processing step can be  implemented through a wide variety of methods \cite{Temme17,Filippov2023}, but  incurs exponential sampling overheads with increasing circuit layers \cite{Berg22,Takagi23,Tsobouchi23,Quek22}.}  Despite the apparent challenges, PEC has been used to significantly extend the accessible circuit depths of a given noisy quantum device, as demonstrated in recent experiments \cite{Berg22}.  Some other examples of quantum error mitigation include zero-noise extrapolation \cite{Temme17,Li17,Kim23}, symmetry-based error detection \cite{Bonet18,McArdle19}, cooling/purification \cite{Cotler19,Huggins21,Koczor21}, and learning based methods \cite{Czarnik21,Srikis21}. \new{ We show in the case of PEC and TEM, that the mitigation threshold  is related to the fundamental behavior of disordered Ising models in statistical physics.  As a result, we can draw on a wide range of previous results from those systems  in analyzing the threshold behavior for these error mitigation strategies. On general grounds, one might expect similar similar types of robustness thresholds to arise for other error mitigation methods, but the explicit mappings to disordered Ising models are less appparent in those cases.}

\new{Based on these connections to disordered Ising models, our central prediction is that PEC and TEM have an error mitigation threshold in random circuits in spatial dimensions $D \geq 2$ when the noise is imperfectly characterized, i.e., there is uncertainty in the parameters of the noise model that have random fluctuations in space and/or time.  }
We consider a model of random unitary circuits subject to depolarizing noise with binary disorder in the depolarization rate [see Fig.~\ref{fig:model}(a)], building on extensive work characterizing quantum phases and phase transitions in hybrid random circuits \cite{FisherRev22,PotterRev22}.  To mitigate the effects of the noise we apply a uniform ``antinoise'' channel that inverts the depolarizing noise on average.  We show that a replica statistical mechanics description of the problem has close parallels to the classical random field Ising model (RFIM) in $D+1$ dimensions. At the zero-mean field condition in the RFIM, a simple heuristic argument, originally due to Imry and Ma, \cite{Imry75} shows that the ordered phase survives random symmetry-breaking terms in sufficiently high dimensions.  In the random circuit problem, this analysis indicates the possibility for thresholds in the robustness of  error mitigation for $D\ge 2$, while $D=1$ is the marginal dimension.  We also analyze the performance of error mitigation as a function of circuit depth.  For depths larger than the linear size of the system, the mitigation fails, resulting in a fidelity worse then one obtains for the maximally mixed state.  The corresponding conjectured phase diagram of the system is shown in Fig.~\ref{fig:model}(b) for $D\ge 2$. Analytic and numerical studies of the two-copy replica theory and exact simulations of the mitigated circuit dynamics are consistent with our predictions.  

\begin{figure}[t]
\begin{center}
\includegraphics[width = .49 \textwidth]{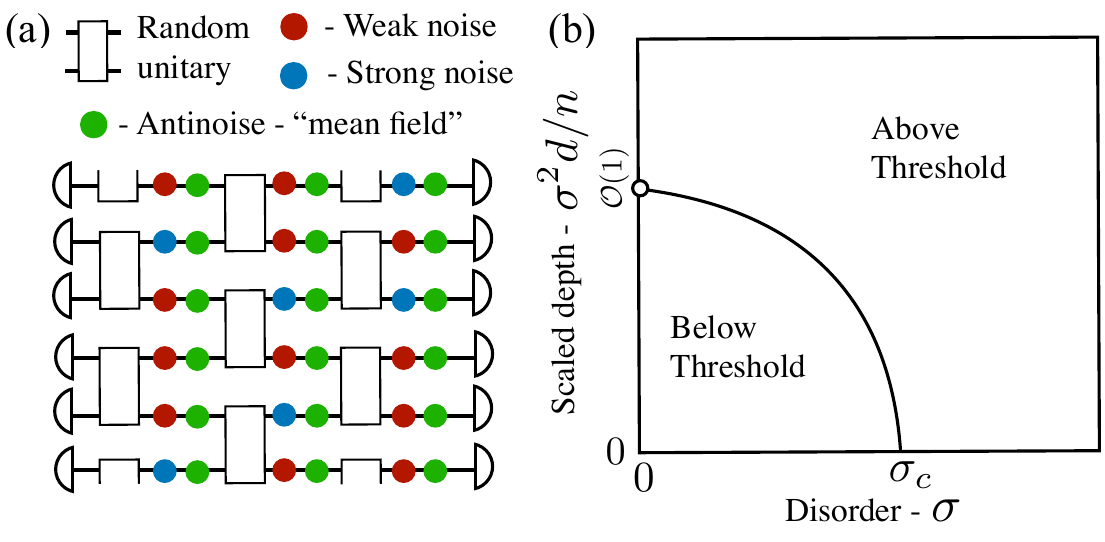}
\caption{Model and phase diagram: (a) Noisy random circuit for qubits with spacetime disorder in the noise rate and error mitigation. State-preparation and measurement are represented by semicircles at the beginning and end of the circuit, respectively.    (b) Phase diagram for the error mitigation threshold in this model for $D\ge 2$ spatial dimensions.  At high disorder rates and linear depths, the system transitions from a ``below-threshold'' phase where  mitigation succeeds on average to an ``above-threshold'' phase where mitigation fails.}
\label{fig:model}
\end{center}
\end{figure}

The paper is organized as follows: In Sec.\ \ref{sec:model}, we introduce our model for error mitigation based on noise and anti-noise channels.  In Sec.~\ref{sec:mft}, we present a simple mean-field theory solution for the error mitigation transition in the case of a Brownian circuit model. In Sec. \ref{sec:statmech}, we introduce the statistical mechanics mapping for discrete circuits that we use to analyze the phase transition in finite-dimensional models.  In Sec.\ \ref{sec:stab}, we present a rigorous result on the absence of an ordered phase in 1D with quenched disorder, as well as simple heuristic arguments for the existence of an ordered phase in more than 1D with spacetime random disorder and more than 2D with quenched disorder.  In Sec.\ \ref{sec:num}, we present numerical simulations of a mutual information probe in the statistical mechanics model and exactly evolved circuits to support our analytical results.  In Sec.~\ref{app:crossing}, we present an alternative probe of the phase transition based on the entropy of a single qubit.  In Sec.~\ref{app:fid}, we present an application of our results to benchmarking quantum advantage in random circuit sampling experiments. We present our conclusions in Sec.~\ref{sec:con}.  In the appendixes, we present further details on the mean-field calculations and the statistical mechanics mapping.


\section{Model}
\label{sec:model}

The basic setup we consider is a quantum circuit model as illustrated in Fig.~\ref{fig:model}(a). A system of qubits is initialized in a product state and a quantum circuit is applied to the system composed of two qubit gates.  Each two-site gate in the circuit is chosen to be Haar random \cite{Nahum16,Nahum17,vonKeyserlingk17} and are interspersed with noise and what we call ``antinoise'' channels (defined below).
The choice of Haar random gates is made purely for convenience of analysis. Our results are expected to hold more generically in systems displaying quantum chaotic dynamics \cite{Zhou2020}.

For the noise channel, we focus on the case of depolarizing noise and its inverse channel as the antinoise
\begin{align}
    \mathcal{E}_{\bm{x},t}(\rho) & = (1-q_{\bm{x},t}) \rho + q_{\bm{x},t} \trace_{\bm{x}}[\rho] \otimes I/2,\\ \label{eqn:anoise}
    \mathcal{A}_{\bm{x}}(\rho) &= \frac{\rho - q_a \trace_{\bm{x}}[\rho] \otimes I/2}{1-q_a},
\end{align}
where $\trace_{\bm{x}}[\rho]$ is a partial trace operation at site $\bm{x}$.  
Tuning $q_a = q_{\bm{x},t}$ inverts the noise at space-time site $(\bm{x},t)$ such that $\mathcal{A}_{\bm{x}}\circ \mathcal{E}_{\bm{x},t}(\rho) = \rho$.  To include the disorder, the noise rates are drawn randomly at each space-time location of the circuit  from one of two values $q_{\bm{x},t}=q_{1(2)}$ with probability $p$ and $1-p$, respectively, and $q_1<q_2$.  The antinoise can be tuned to a zero mean-field condition by taking $(1-q_a) = (1-q_1)^p (1-q_2)^{1-p}$ equal to a geometric mean. This condition ensures that a Pauli operator $\sigma_{\bm{x}}^\mu$ ($\mu \in \{X,Y,Z\}$) evolved under multiple rounds of noise and antinoise up to depth $t = d$ will have  expectation value
\begin{equation}
 \trace\bigg[\sigma_{\bm{x}}^\mu \prod_{t=1}^d \mathcal{A}_{\bm{x}} \circ \mathcal{E}_{\bm{x},t}(\rho)\bigg] = \prod_t \frac{1-q_{\bm{x},t}}{1-q_a} \trace[ \sigma_{\bm{x}}^\mu \rho],
\end{equation}
where the prefactor tends to a log-normal distribution that concentrates around one.  Any deviations of $q_a$ from this zero mean-field condition will result in an exponential growth rate that scales linearly in the depth $d$.  At the zero-mean field condition, however, the noise is canceled on average and the prefactor deviates from one only due to random fluctuations as $e^{\pm \mathcal{O}(\sqrt{d})}$ (i.e., the exponential growth rate is suppressed by a factor $\sqrt{d}$).

To see how unitary dynamics can improve the robustness of error mitigation, it is helpful to consider a toy model consisting of two sites $L/R$ with $q_L < q_R$ fixed in time. Now, take $1-q_a = \sqrt{(1-q_R)(1-q_L)}$ to satisfy the zero mean-field condition. The $R$ site will be subject to excess noise and the $L$ site will experience excess antinoise.   The case of weak unitary dynamics is analogous to the above threshold phase.  In the limit where the unitary dynamics tends to trivial evolution, any state initialized with non-zero Pauli expectation value on site $L$ will experience exponential growth of the expectation values leading to an unphysical density matrix after an $\mathcal{O}(1)$ time.  We discuss below how this type of rapid instability to unphysical states is characteristic of the above threshold phase. 

To build in the effects of unitary dynamics, consider the base circuit to consist of repeated applications of SWAP gates.  In this case, it is readily verified that the effects of the noise on Pauli expectation values will be completely canceled for arbitrarily deep circuits except for the most recent layer of gates.  In this two-qubit example, Haar random  gates behaves similarly to SWAP gates in that the noise is cancelled on average.  However, the temporal randomness will lead to exponential fluctuations of Pauli expectation values of size $e^{\mathcal{O}(\sqrt{d})}$, where the term in the exponent will follow the statistics of a random walk.  This tendency of the random unitary dynamics to suppress effects of the noise/antinoise from linear in $d$ to $\sqrt{d}$ scaling for certain initial conditions is characteristic of the behavior we find for the below-threshold phase.  

\section{Mean-field theory}
\label{sec:mft}

To develop some basic intuition for the threshold behavior in discrete circuits, we now consider a simple mean-field theory of the transition formulated in continuous time.
In this approach, the noise and antinoise terms are captured through a Lindblad master equation \cite{Sun21}
\begin{align}
    \dot{\rho} &=  -i [H(t),\rho] -\sum_{j}(\gamma_j - \gamma_a) \big( \rho - \trace_j[\rho] \otimes  I/2\big),\\
H&(t) = \sum_{i<j,\mu,\nu}J_{ij\mu\nu}(t)\sigma_i^\mu \sigma_j^\nu,
\end{align}
where $H(t)$ is a time-dependent Hamiltonian that sets the unitary evolution, $\gamma_i$ is the random noise rate on site $i$, and $\gamma_a$ is the antinoise rate.  We see that antinoise simply acts as a negative decay rate in the master equation.  In this formulation it is particularly clear that setting $\gamma_a = \gamma_j$ for all $j$ will completely remove the effects of noise.

We focus on the disordered case, however, where the noise rates are not known perfectly to the antinoise.  In particular, $\gamma_i =\gamma_{1(2)}$ is drawn randomly with probability $p$ and $1-p$, respectively. The noise is canceled on average by  setting $\gamma_a = \bar{\gamma} \equiv p\,\gamma_1 +(1-p)\,\gamma_2$. We take the unitary evolution to be given by an all-to-all coupled Brownian circuit model in which the $J_{ij\mu\nu}(t)$ are drawn from a white-noise Gaussian random process with variance parameters \cite{Lashkari2013,Jian22}
\begin{equation}
\langle J_{ij\mu\nu}(t) J_{k \ell \gamma \delta}(t') \rangle = \frac{J}{2 N} \delta_{ik} \delta_{j\ell} \delta_{\mu \gamma} \delta_{\nu \delta} \delta(t-t'),
\end{equation}
where $N$ is the number of qubits.  

To treat this model analytically, we average the dynamics  over the random coupling constants for multiple identical copies of the density matrix $\rho(t)^{\otimes k}$ ($k$ is a replica index).  In the equations of motion for the averaged moments $\rho_k(t) \equiv \mathbb{E}_H[\rho ^{\otimes k}]$, the unitary drive-term averages to zero, leading to a purely dissipative master equation for $\rho_k(t)$ (see Appendix \ref{app:mft}). Focusing on second moments $(k=2)$ gives rise to a particularly simple model with two global product steady states given by $ I^{\otimes N}$ and $S^{\otimes N}$ \cite{Jian22}, where $S$ is the SWAP operator acting on the two copies.  The all SWAP steady-state captures non-trivial corrections to second moment observables like the purity of sub-systems.  In the mean-field approximation, we enforce the density matrix to take a product state form $\rho_2 = \bigotimes_{i=1}^N \sigma_i$.  
Due to the lack of symmetries or conservation laws in the problem, $\sigma_i$ can be parameterized as
\begin{equation}
\sigma_i = (1/4+\delta_i) |s\rangle \langle s| + (1/4 -  \delta_i/3) P_T,
\end{equation}
where $\delta_i$ is the deviation from an infinite temperature state, $|s\rangle$ is a two-qubit singlet state  across the two copies, and $P_T$ is the projector onto the two-qubit triplet subspace of the two copies.   

The mean-field equations for $\delta_i$ take the form
\begin{equation}
\dot{\delta}_i = -4 \Big[  \Delta_i + \frac{ J}{N} \sum_{j \ne i} (3 + 4\delta_j) \Big] \delta_i,
\end{equation}
where $\Delta_i = \gamma_i - \gamma_a$.  This simple equation captures much of the basic physics we find later in the random circuit models, including the presence of a phase transition.  Notably, we can see that at weak values of the disorder, the interaction term acts as a restoring force towards the fixed point $\delta_i = 0$.  The other fixed point $\delta_i = -3/4$ corresponds to an unphysical density matrix, which is unstable at weak disorder.  At a critical disorder strength of $|\Delta_1| = 3J$ (see App.~\ref{app:mft}), however, the $\delta_i = 0$ fixed point becomes unstable.  The dynamics flows to a new stable fixed point which is an unphysical density matrix.  Thus, we see that the above threshold phase in the mean-field theory is characterized by an instabilty of the density matrix to unphysical states.  

To rigorously characterize the threshold behavior, particularly in finite-dimensional systems, we need a  more systematic treatment of the correlations in the state beyond the mean-field approximation.  To develop this approach, we turn to the discrete circuit models.  An added benefit of this formulation is that the connections to the RFIM become more explicit.

\section{Statistical Mechanics Mapping}
\label{sec:statmech}

To analyze the problem more rigorously, we turn to well-established mappings from random quantum circuit dynamics in $D$ spatial dimensions to statistical mechanics models in $D+1$ dimensions \cite{Nahum17,vonKeyserlingk17,Jian19,Bao20,Dalzell2022,Dalzell2021,Deshpande22,Gao21}.  We outline the details of the mapping in  Appendix \ref{app:statmech}.  The approach we take is to expand $\rho_k \equiv \mathbb{U}[\rho^{\otimes k}]$ into a basis of operators based on representations of the permutation group $S_k$ \cite{Dalzell2022}.  One can then derive update rules for the averaged state following each layer of gates, noise, and antinoise.  In this paper, we focus on the two-replica case where the two operators are $ I$ and $S$ (the SWAP gate).  We expect this model to be sufficient to capture many of the properties of the transition, as has been argued in the unitary case \cite{Nahum17}.    Numerical evidence based on exact simulations is presented below that support this claim.

For a system of $N$ qubits, without noise/antinoise, $\rho_2$ has two fixed points $ I^{\otimes N}$ and $S^{\otimes N}$ just as in the Brownian circuit model.  Configurations with persistent ``domain walls'' (i.e., mixtures of $I$ and $S$) are exponentially damped with depth in the model.  As a result, the two fixed points can be interpreted as ferromagnetic ground states of an Ising-type spin model. In this interpretation, noise acts as a symmetry breaking field because it favors the identity configuration and damps out configurations with $S$ operators \cite{Bao20}.  Antinoise acts as a field in the opposite direction, but it also adds additional sign-structure to the problem from the second term in Eq.~\eqref{eqn:anoise}. We account for this sign by labeling spacetime configurations of $I$ and $S$ operators by the sign of their contribution, then we can write $\rho_2 = \rho_2^+ - \rho_2^-$, where $\rho_2^\pm$  sum over spacetime configurations with a positive/negative sign.  

\begin{figure*}[t!]
\begin{center}
\includegraphics[width = .8 \textwidth]{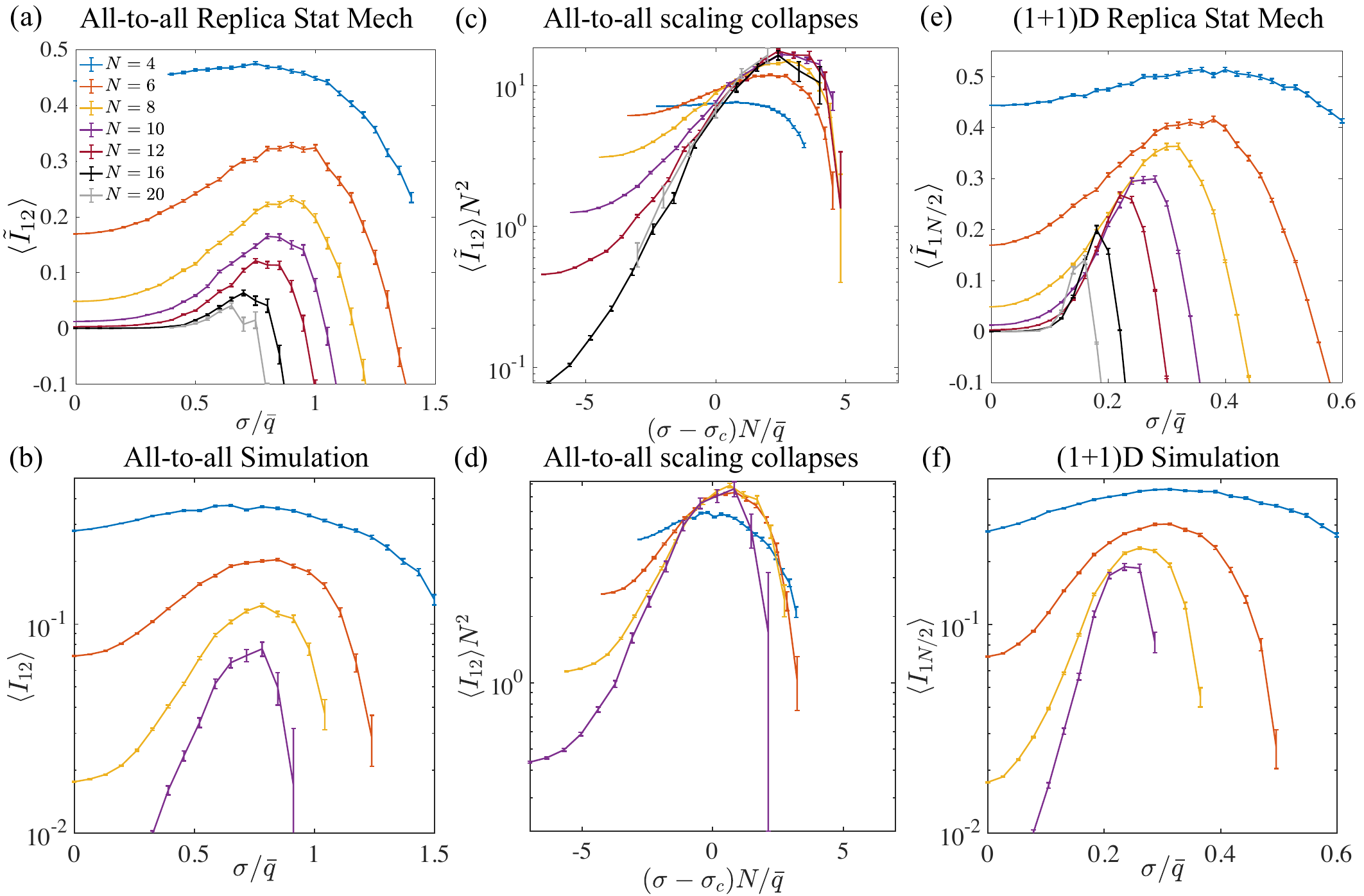}
\caption{Numerical tests of the threshold behavior: (a) Numerical simulations of circuit and noise-disorder averaged the correlation metric $\tilde{I}_{12}$ obtained from purity of subregions for two random sites in the system.  (b) Mutual information $I_{12}$ for exact simulations of the mitigated noisy random circuit.  (c-d) Scaling collapse of data from (a-b) using $\sigma_c/\bar{q} = 0.65(5)/0.55(5)$ and critical exponent on $x$-axis obtained from local probe (see Appendix \ref{app:crossing}).  The critical scaling on the $y$-axis is an ansatz.  (e-f) Similar quantities as (a-b) for antipodal sites of a 1D chain.  The strong drift in the peak with increasing system size is consistent with the lack of a below-threshold phase in this model.  In all the plots, we took a Haar random initial pure state with noise parameters $\bar{q}=0.1$ and $p = 0.9$ and ran to depth $d=4N$.}
\label{fig:data}
\end{center}
\end{figure*}

\section{Stability of Ordered Phase to Disorder}
\label{sec:stab}

Before proceeding, we first review the basic Imry-Ma argument that illustrates how an ordered ferromagnet can be stable to disorder for $D>2$ dimensions  \cite{Imry75}.  The classical Hamiltonian takes the form $H = -\sum_{\langle ij\rangle} \sigma_i \sigma_j + \sum_i h_i \sigma_i$, where $\sigma_i \in \{+1,-1\}$ are classical Ising spins and $h_i$ is drawn from an independent random distribution on each site. At the zero mean-field  condition,  $\langle{h_i}\rangle=0$,  the energy cost for ordering a region of size $L$ into a ferromagnetic domain in the presence of the random field will scale as $\sigma L^{D/2}$ where $\sigma$ is the standard deviation of the disorder.  On the other hand, the energy cost associated with creating a domain wall of broken ferromagnetic bonds at the boundary will scale as $L^{D-1}$.  We see that for $D >2$ the domain wall cost is dominant and the ordered phase is stable to small amounts of disorder.  However, for $D<2$, the random field wins and destroys the ferromagnet.  In $D=2$, it was eventually proven rigorously that the ordered phase is also unstable \cite{Aizenman89}.

To make precise connections of the error mitigation threshold problem to Imry-Ma arguments, we analyze finite-dimensional versions of the statistical mechanics models.  First, we define a \textit{simple initial condition} as a state $\rho_2$ that is identity everywhere except in a contiguous region $A$ where it is given by $\trace_{A^c}[\rho_2] \propto I^{\otimes |A|}+S^{\otimes |A|}$.  Such an initial condition can be prepared physically by taking a Haar random state (or a state drawn from a two-design such as the Clifford group \cite{DiVincenzo02}) on $A$ and the infinite temperature state on the complement of $A$.  The first model we consider is a \emph{quenched 1D} model described by a one-dimensional brickwork circuit of the type shown in Fig.~\ref{fig:model}(a) with Haar random two-qubit gates and spatially random noise rates that are constant in time.  We define $\sigma = \sqrt{p(1-p)}(q_2-q_1)$ as the standard deviation of the noise rates.  We prove that the quenched 1D model always exhibits an instability for one of the simple initial conditions. 

\begin{thm}[Instability in quenched 1D]
For the quenched 1D model with any $\sigma >0$, there is a simple initial condition for subregion $A$ of size $\mathcal{O}(\log N)$ such that as $N\to \infty$ and $d\to \infty$ with high probability $\log ||\rho_2^\pm || \ge e^{\Gamma d }/2^{2N-|A|}3^{|A|}$ for a constant $\Gamma > 0$.
\label{thm:quenched-1d}
\end{thm}
\begin{proof}
The proof of the theorem essentially follows a standard Imry-Ma argument.  For any finite $\sigma$ as $N\to \infty$ with high probability (whp) there will be a contiguous region of sites $A$ with noise rate $q_1$ of size $|A|=\mathcal{O}(\log N)$. For sufficiently large $N$ there will be a region satisfying $(1-q_1)^{|A|}/(1-q_a)^{|A|}  > (5/2)^2$ whp, which implies that the $\mathcal{O}(1)$ boundary cost of the two domain walls at the edge of $A$ can be overcome by the amplification of the $S$ configurations in the bulk. Let $m$ be the minimum integer such that $(1-q_1)^m/(1-q_a)^m  > (5/2)^2$ and define $\Gamma = \ln [(1-q_1)^m/(1-q_a)^m]- 2\ln 5/2$ Taking a configuration of all $S$ on $A$ with identities elsewhere that moves by order 1 site at the boundary at each time-step according to the rules for the replica model provides a contribution to $\rho_2^{\pm}$ that diverges faster than  $e^{\Gamma d}/2^{2N-|A|}3^{|A|}$.  The denominator is the normalization constant for the initial state.
\end{proof}

As a corollary to this result, an initial pure product state evolved with this model will also have this instability.  The converse of this type of Imry-Ma argument is that we expect the instability of these initial conditions to be suppressed for weak enough disorder for high enough dimensions.  In this case, for large regions $A$ the boundary cost of the domain wall scales as $|A|^{(D-1)/D}$, whereas the amplification from the antinoise scales as $|A|^{1/2}$.  As a result, with quenched disorder in time, we expect the instability to vanish for $D\ge 3$ at sufficiently small values of $\sigma$, while $D=2$ is the marginal dimension.  For space-time random disorder, the marginal dimension is $D=1$ and an ordered phase persists for $D\ge 2$.  
Note, that at large values of $\sigma$, the domain wall cost can be overcome even by a single site with strong antinoise, leading to an instability.  As a result, we expect to find a disorder driven phase transition at finite $\sigma$ in sufficiently high-dimensions associated with the onset of this instability.

\section{Mutual Information Probe}
\label{sec:num}

Figure \ref{fig:data} summarizes our numerical results on the mutual information between two sites in the system, which serves as a probe of the two phases.  We run numerical simulations of random circuits, together with noise and antinoise, to characterize this phase transition in two settings.  In the first case [see Fig.~\ref{fig:data}(a-d)], we simulate an infinite-dimensional circuit where entangling gates can exist between any pair of qubits and noise is random in both space and time.  In the second case [see Fig.~\ref{fig:data}(e-f)], we simulate $(1+1)D$ circuit dynamics assuming quenched noise, where the randomness is only spatial. The first model is a test of the high-dimensional limit of the model, while the second model is the same one considered in Theorem \ref{thm:quenched-1d}.

In both cases, we initialize a state of a $N$-qubit register in a global Haar-random state, which corresponds to the string $\propto I^{\otimes N} + S^{\otimes N}$ in the statistical-mechanical mapping.  For the infinite-dimension model, we apply random Haar gates on a set of $N/2$ disjoint pairs, randomly drawn from all $N(N-1)/2$ possible pairs of qubits. For the $(1+1)D$ model, each time-step corresponds to a layer in the brickwork circuit in Fig.~\ref{fig:model}(a) with periodic boundary conditions. Each layer of unitary gates is followed by the noise-antinoise channel on all qubits. 

After evolving the circuit to depth $d$, we take the partial trace of two qubits and calculate correlation metrics between the two qubits \cite{Skinner2019,Li2019}.  These low-order correlations are exponentially suppressed in $N$ for $\sigma = 0$, but they develop a peak at the critical point where long-range correlations develop. For $D=1$, the pair consists a randomly chosen qubit and its corresponding anti-pode at a distance of $N/2$; for the all-to-all model, both qubits are randomly chosen. 

We study the two systems using the replica statistical-mechanical mapping introduced above, with background details in  Appendix \ref{app:statmech}. The mapping lets us probe the phase transition for system sizes as large as $N=20$. For sizes up to $N=10$, we complement the numerics on statistical-mechanical mapping with direct circuit simulations.  In the replica stat-mech model, we study a measure of correlations $\tilde{I}_{ab}= - \log_2 \mathbb{E}_U[\mathrm{Tr}[\rho_a^2]] - \log_2 \mathbb{E}_U[\mathrm{Tr}[\rho_b^2]] + \log_2 \mathbb{E}_U[\mathrm{Tr}[\rho_{ab}^2]] $ averaged over noise-disorder (denoted by $\langle \cdot \rangle$ brackets), while in the exact simulations we study the standard mutual information $I_{ab}= S_a + S_b - S_{ab}$ for the von Neumann entropy $S_a = -\trace[\rho_a \log_2 \rho_a]$ \cite{NielsenChuang}.

The numerical results provide strong supporting evidence for the theoretical scenario outlined above.  We see evidence for a stable ordered phase in the high-dimensional limit of the all-to-all model, whereas the 1D model with quenched randomness in the noise is consistent with the lack of an ordered phase.  In addition, we estimate critical exponents for the transition in Fig.~\ref{fig:data}(b-c) that are consistent with simple rational values characteristic of mean-field behavior.

\section{Local Probe}
\label{app:crossing}
In Fig.~\ref{fig:data}, we presented different correlation metrics between two qubits in the system as a probe of the phase transition.  Here, we present an alternative metric based on the entropy of a single qubit in the system.  

Deep in the below threshold phase, this single-qubit entropy quantity saturates to one bit, while above threshold it diverges to large negative values due to the unphysical density matrix.  As a result, we expect a crossing to occur at the phase transition.  Numerical simulations of the two-replica stat-mech model for the all-to-all circuit model illustrate this behavior.  In Fig.~\ref{fig:crossing}(a), we show the unscaled behavior of the entropy for different system sizes, which shows a crossing near $\sigma_c/\bar{q} = 0.65(5)$.  Collapsing the data with this value of $\sigma_c$ fixed, we estimate a critical exponent $\mu = 1.0(2)$.  In Fig.~\ref{fig:data}(c-d), we fixed $\mu = 1$ in collapsing the mutual information data based on these scaling results.  This single-qubit quantity also has advantages for experimental probes of the transition as it requires minimal tomographic overhead to estimate. 

\begin{figure}[h!]
\begin{center}
\includegraphics[width = .49 \textwidth]{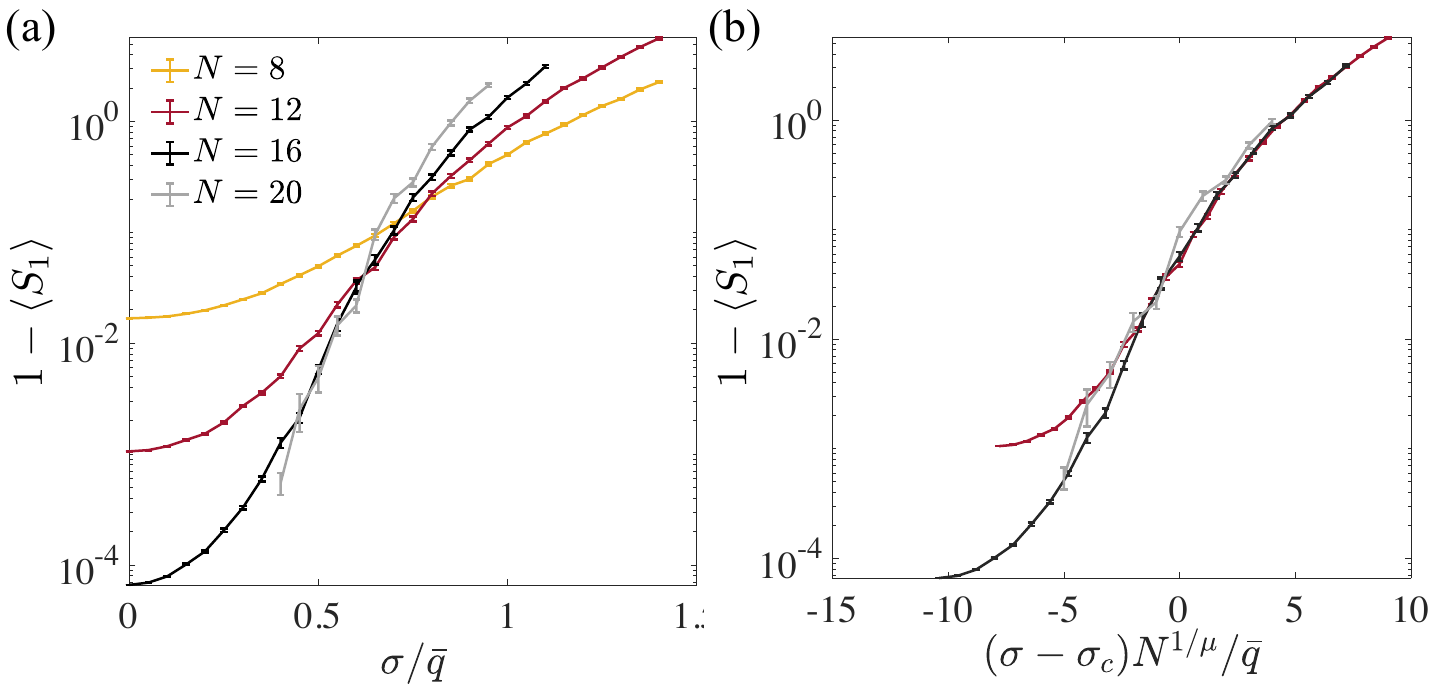}
\caption{Local probe of the error mitigation threshold: (a) Entanglement entropy of a single-site in the system for different system sizes.  The larger sizes begin to develop a crossing, indicative of the phase transition.  (b) Scaling collapse for $\sigma_c/\bar{q} = 0.65(5)$ and $\mu = 1.0(2).$ }
\label{fig:crossing}
\end{center}
\end{figure}

\section{Error Mitigated Fidelity Benchmarks}
\label{app:fid}

Here, as an application of our results, we introduce mitigated fidelity benchmarks and demonstrate an exponential improvement in a mitigated verison of the linear cross-entropy benchmark below the error mitigation threshold.

The task of sampling from the output distribution of a noiseless random circuit is widely conjectured to be intractable with classical computers \cite{Boixo2018,Bouland2019,Movassagh2018,Hangleiter22}.  This conjecture forms the basis for claims of achievement of quantum computational advantage in recent experiments \cite{Arute2019a,Wu21,Zhu21}.  The experimental claims remain controversial, however, partly because of the effects of noise on the output that render the signal classically simulatable at high depth \cite{Gao2018,Bouland2021,Aharonov22}. 
To provide evidence that the output signal still remains close to the ideal case, one can estimate fidelity benchmarks using the samples from the experiment.  Verifying the claim of computational advantage in the case of noisy circuits then reduces to the task of achieving a sufficiently high ``score'' on the benchmark \cite{Aaronson2019}.

Recall that the linear XEB is  given by the formula \cite{Boixo2018,Arute2019a}
\begin{equation}
    F_{\rm XEB} = 2^n \sum_x p_{ n}(x) p_{0}(x) -1,
\end{equation}
where $p_0(x) = |\langle x | U_d \cdots U_1 | 0 \rangle|^2$ is the probability of measuring outcome $x$ for the noiseless circuit of depth $d$ and $p_{ n}(x) = \langle x| \mathcal{E}_d \circ \mathcal{U}_d \circ \cdots \mathcal{E}_1 \circ \mathcal{U}_1( | 0 \rangle \langle 0 |) | x \rangle $ is the analogous probability for the noisy circuit.

\begin{figure}[tb!]
\begin{center}
\includegraphics[width = .48 \textwidth]{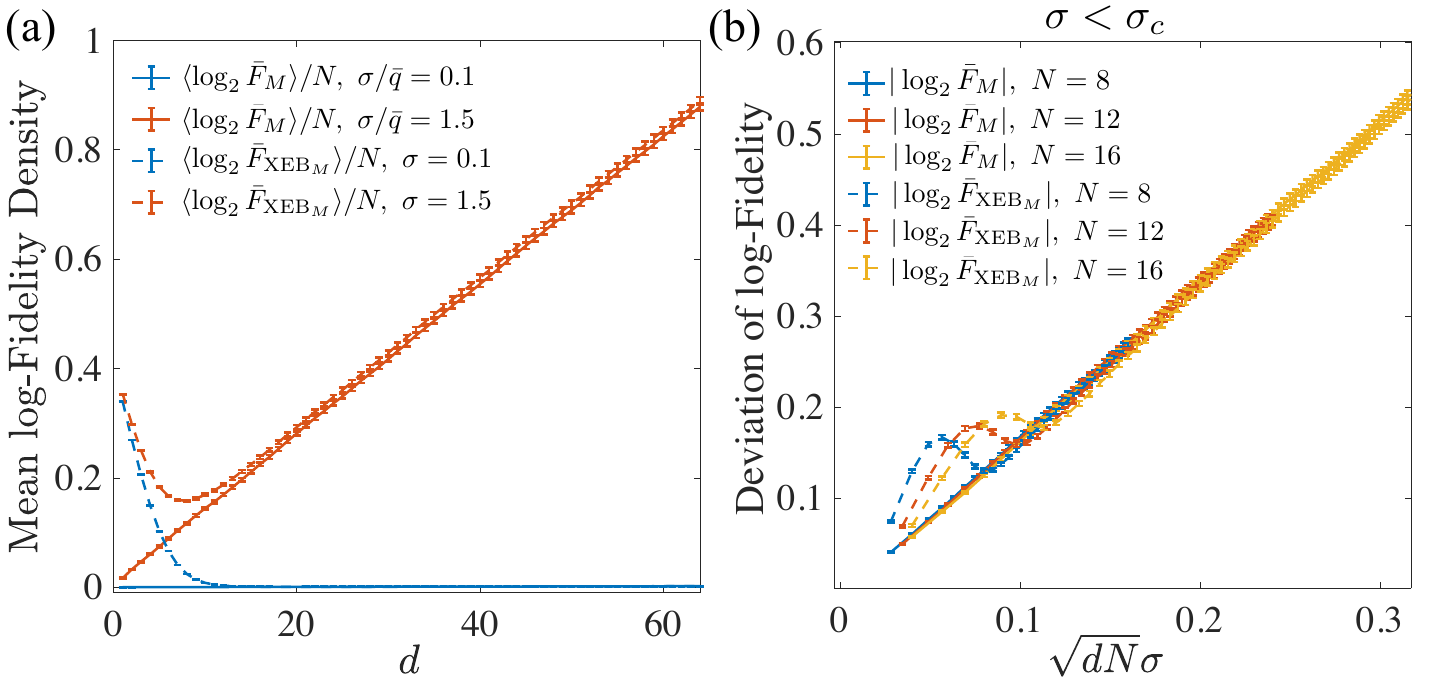}
\caption{Dynamics of fidelity benchmarks above and below threshold: (a) Disorder-averaged mean value of the logarithm of the circuit averaged mitigated fidelity $-\langle \log \bar{F}_M\rangle/N$ and cross-entropy benchmarking mitigated fidelity $-\langle\log F_{{\rm XEB}_M}\rangle/N$ above and below the error mitigation threshold in the all-to-all model.  (b) Dynamics of the standard deviation over the noise of the log-circuit-averaged mitigated fidelities below threshold showing the improved scaling of the typical log-fidelity as $\pm\mathcal{O}(\sqrt{Nd})$.  In both plots, we took a pure product initial state with noise parameters $\bar{q}=0.1$ and $p= 0.9$.}
\label{fig:fidelity}
\end{center}
\end{figure}

We now introduce the mitigated linear XEB, which is instead given by the formula
\begin{equation}
    F_{\rm XEB_M} = 2^n \sum_x  p_{n}(x) p_{a}(x)-1,
\end{equation}
where $p_a(x) = \langle x| \mathcal{A}\circ \mathcal{U}_d \circ \cdots \mathcal{A} \circ \mathcal{U}_1( | 0 \rangle \langle 0 |) | x \rangle $ is a quasi-probability for the circuit with antinoise inserted in place of noise.   The quantity $p_a(x)$ can be computed on a classical computer, which leads to a sampling formula for $F_{\rm XEB_M}$ using $M$ samples $x_i$ obtained from $p_n(x)$
\begin{equation}
    F_{\rm XEB_M} = \frac{2^n}{M} \sum_{i=1}^M p_a(x_i) - 1.
\end{equation}
This formula illustrates that the mitigated fidelity can be obtained without directly implementing PEC except in purely classical post-processing.
After circuit averaging, one can quickly show that for depolarizing noise and its antinoise partner, the mitigated fidelity is equivalent to the formula 
\begin{equation}
    \bar{F}_{{\rm XEB}_M} = 2^n \mathbb{E}_U \sum_x  p_{0}(x) p_{an}(x)-1,
\end{equation}
where $p_{an}(x) =  \langle x| \mathcal{A}\circ \mathcal{E}_d \circ \mathcal{U}_d \circ \cdots \mathcal{A} \circ \mathcal{E}_1 \circ \mathcal{U}_1( | 0 \rangle \langle 0 |) | x \rangle$ implements the antinoise on the same copy as the noise.  This identity follows because the noise and antinoise on one copy have the identical effect on the $ I$ and $S$ operators after averaging over circuits.  From this expression, we see that, in the case of perfect mitigation, $\bar{F}_{{\rm XEB}_M}$ reduces to its ideal value.

We can also define a mitigated fidelity that takes the form 
\begin{align*}
    F_M &= \tr[ \mathcal{A} \circ \mathcal{U}_d \circ \cdots \circ \mathcal{A} \circ \mathcal{U}_1(|0\rangle \langle 0|) \\
    &\times \mathcal{E}_d \circ \mathcal{U}_d \circ \cdots \circ \mathcal{E}_1 \circ \mathcal{U}_1(|0\rangle \langle 0|)].
\end{align*}
In the case where $\mathcal{E}_i = \mathcal{A}^{-1}$ for every $i$, we can see that $F_M=1$.

In Fig.~\ref{fig:fidelity}], we show that the log-fidelity at low noise rates and $d$ grows as $\mathcal{O}(Nd)$, whereas using error mitigation this scaling can be improved to $\mathcal{O}(\sqrt{Nd})$, representing an exponential improvement in the score that brings it closer to the $\mathcal{O}(d)$ scaling of the log-total variation distance \cite{Deshpande22}.  Moreover, as explained above, the cross-entropy benchmark fidelity can be mitigated entirely in classical post-processing. As a result, the mitigated XEB fidelity can be estimated with existing experimental data from random circuit sampling experiments. 

\section{Conclusions}
\label{sec:con}

In this work, we showed how imperfect noise characterization in the presence of random spacetime fluctuations in the noise parameters can lead to thresholds in the robustness of error mitigation.  We established the absence of a robust phase rigorously in the case of one-dimensional noisy random quantum circuits with quenched spatial disorder.  For higher dimensional systems, we presented analytic and numerical evidence that the below threshold phase persists for weak disorder and short enough times.  To further support our arguments, we presented a mean-field theory where the phase transition manifests in an instability of the mitigated density matrix to unphysical states.  Overall, our numerical and analytical results are consistent with the conjectured phase diagram in Fig.~\ref{fig:model}(b), but further work remains to fully validate the depth-dependence of the phase boundary due to finite-size effects in the numerics.

To connect our results to practical considerations in the implementation of error mitigation, it is important to more carefully analyze the sample complexity required to see our predicted effects on realistic systems.  The overhead of error mitigation can lead to an exponential sampling cost in the worst case \cite{Berg22,Tsobouchi23,Takagi23,Quek22}, although in near-term devices the added cost can be reasonably overcome on relevant system sizes \cite{Kim2023}.  One of the most promising techniques for PEC is to use tensor-network based approaches where the error mitigation can be done entirely in post-processing without the need to run additional circuits \cite{Filippov2023}.  This approach to error mitigation is particularly well-suited to analog quantum simulation platforms and the exploration of the threshold effects reported here.  We leave a more detailed analysis of practical implementations of our proposal for future work.

The phase transition we discuss only emerges precisely in the thermodynamic limit of large $N$.  In this case, a natural concern is that the exponential cost of error mitigation for circuits with depth greater than $\log \log N$ would make the phase transition studied here unphysical \cite{Tsobouchi23,Takagi23,Quek22}.  However, in more than two dimensions, this phase transition  also arises in constant depth circuits, as we showed directly in the case of Brownian circuits.  In that case, the phase transition manifests as the emergence of an instability of the mitigated density matrix to unphysical states at large disorder.

An important application of our results is to improved benchmarking and verification of random circuit sampling in noisy devices \cite{Boixo2018,Arute2019a,Bouland2019,Movassagh2019,Wu21,Zhu21}.
In  Sec.~\ref{app:fid}, we demonstrated an exponential improvement in a mitigated version of the so-called ``linear cross-entropy'' fidelity below the error mitigation threshold.  Notably, this mitigated benchmark can be computed using existing experimental data.  The below threshold behavior  may have important consequences for proofs of quantum computational advantage in noisy circuits \cite{Aaronson2019,Bouland2021,Deshpande22}, potentially motivating the inclusion of error mitigation in complexity theoretic arguments.  Moreover, the exponential improvement in the fidelity score for $D\ge 2$ might indicate a method to foil ``spoofing'' algorithms that have been developed for unmitigated fidelity benchmarks in noisy random circuits \cite{Gao21}.

These results also have implications for measurement-induced phase transitions in hybrid random circuits \cite{Skinner2019,Li2019,FisherRev22,PotterRev22}.  Such models are not fault-tolerant because circuit-level depolarizing noise drives the system to a short-range correlated mixed state (it acts as a symmetry breaking field in the replica theory \cite{Bao20}).  The addition of the antinoise terms considered here can restore the symmetries of the model on average, potentially allowing for a stable ordered phase.  

More broadly, the existence of error mitigation thresholds has important implications for near-term quantum simulation and quantum computing.  
An error mitigation threshold qualitatively similar to the one studied in this work should arise for any chaotic quantum dynamics including Hamiltonian evolution with sufficiently strong interactions relative to the dissipation.  

\begin{acknowledgements}
We thank Chris Baldwin, Gregory Bentsen, Abhinav Deshpande, Bill Fefferman, Alexey Gorshkov, David Huse, Zlatko Minev, Alireza Seif, Sagar Vijay, and Brayden Ware for helpful discussions.  Work supported in part by NSF QLCI grant OMA-2120757, NSF DMR-1653271, and NSF grant PHY-1748958.
\end{acknowledgements}

 \section*{Data Availability}
The data that support the findings of this article are openly available \cite{data}.

\appendix 


\section{Mean-Field Theory in Brownian Circuits}
\label{app:mft}
Here, we develop a mean-field theory for the error mitigation threshold in the noisy-mitigated Brownian circuit model.

The unitary dynamics in the Brownian quantum circuit model is described by a stochastic Hamiltonian
\begin{equation}
H(t) = \sum_{i < j,\mu,\nu} J_{ij\mu\nu}(t) \sigma_i^\mu \sigma_j^\nu,
\end{equation}
where $\sigma_i^\mu$ are Pauli operators $\mu \in \{ X,Y,Z\}$ for site $i$ of $N$ qubits and $J_{ij\mu\nu}(t)$ is a white-noise correlated coupling with variance \cite{Lashkari2013,Jian22}
\begin{equation}
\langle J_{ij\mu\nu}(t) J_{k \ell \gamma \delta}(t') \rangle = \frac{J}{2 N} \delta_{ik} \delta_{j\ell} \delta_{\mu \gamma} \delta_{\nu \delta} \delta(t-t').
\end{equation}
The noise and antinoise are treated using a Lindblad master equation
\begin{equation}
\dot{\rho} = -i [H(t),\rho] + \sum_i \frac{\gamma_i - \gamma_a}{4}\ \Big( -3 \rho + \sum_\mu \sigma_i^\mu \rho \sigma_i^\mu  \Big),
\end{equation}
where $\gamma_{i}$ is the local random noise rate and $\gamma_a$ is the antinoise rate.

To analyze the dynamics we derive an effective master equation that describes the replicated system
\begin{equation}
M_k(\rho) = \mathbb{E} [ U^{\otimes k} \rho U^{\dag \otimes k} ],
\end{equation}
where $U = \mathcal{T} e^{-i \int_0^t dt' H(t')}$ is the time evolution operator under the Hamiltonian. \new{ Expanding to second order in an infinitesimal time-step and regrouping terms appropriately, we arrive at the equation $M_k(\rho) = e^{\mathcal{L}_k t}$ for a Lindbladian $\mathcal{L}_k$ given by}
\begin{equation}
\mathcal{L}_k(\rho) = \frac{J}{2 N} \sum_{i \ne j, \mu,\nu,r,s}  \sigma_{ir}^{\mu }\sigma_{jr}^{\nu} \rho \sigma_{is}^{\mu }\sigma_{js}^{\nu } -\frac{1}{2} \{ \sigma_{ir}^{\mu}  \sigma_{is}^{\mu } \sigma_{jr}^{\nu } \sigma_{js}^{\nu }, \rho \} ,
\end{equation}
where $r$ and $s$ are replica indices that run over $1,\ldots,k$.  Thus, we arrive at a master equation describing the average dynamics of the replicated density matrices
\begin{equation}
\dot{\rho} = \mathcal{L}_k(\rho) + \sum_i \frac{\gamma_i - \gamma_a}{4}\ \Big( -3 k \rho + \sum_{r,\mu} \sigma_{ir}^\mu \rho \sigma_{ir}^\mu  \Big).
\end{equation}

To develop the mean field theory, we study the two-replica problem $k=2$. Similar to a  Haar random circuit, the Lindbladian $\mathcal{L}_2$ has two steady states $ I^{\otimes N}$ and $S^{\otimes N}$.  
As our mean-field ansatz, we therefore use a product state of the form $\rho = \bigotimes_{i=1}^N \rho_i$.  A further simplification arises from the nature of the dynamics that has an effective SU(2) symmetry in the average-replica dynamics.  As a result, we can express 
\begin{equation}
\rho_i = (1/4+\delta_i) |s\rangle \langle s| + (1/4 -  \delta_i/3) P_T,
\end{equation}
where $\delta_i$ is the deviation from an infinite temperature state, $|s\rangle$ is a two-qubit singlet state  across the two replicas, and $P_T$ is the projector onto the two-qubit triplet subspace of the two replicas.  

The mean-field equations for $\delta_i$ obtained by inserting our mean field ansatz for $\rho$ into the master equation take the form
\begin{equation}
\dot{\delta}_i = -4 \Big[  \Delta_i + \frac{ J}{N} \sum_{j \ne i} (3 + 4\delta_j) \Big] \delta_i,
\end{equation}
where $\Delta_i = \gamma_i - \gamma_a$.  This equation has the two steady-state solutions $\delta_i =0$ and $\delta_i = -3/4$, corresponding to the $ I^{\otimes N}$ and $S^{\otimes N}$ solutions, respectively.  In the case of binary disorder, we  define two populations of sites $A_{1/2}$ such that $\Delta_i = \gamma_{1/2} - \gamma_a$, respectively, for $i \in A_{1/2}$.  We have the zero-mean field condition $p \Delta_1 +(1-p) \Delta_2 = 0$. Defining $G_{\pm} = \frac{1}{N} \Big( \sum_{i \in A_2} \delta_i \pm \sum_{i \in A_1} \delta_i \Big)$, we arrive at simple closed set of equations in the large-$N$ limit
\begin{align*}
\dot{G}_+ & = - 4J (3 + 4 G_+) G_+ + \frac{ 4|\Delta_1|}{2(1-p)} [G_- - (2p -1) G_+ ], \\
\dot{G}_- & = - 4J (3 + 4 G_+) G_- + \frac{ 4|\Delta_1|}{2(1-p)} [G_+ - (2p -1) G_- ].
\end{align*}
Setting $p = 1/2$, these equations reduce to the particularly simple form
\begin{align}
\dot{G}_+ & = - 4J (3 + 4 G_+) G_+ + 4|\Delta_1| G_- , \\
\dot{G}_- & = - 4J (3 + 4 G_+) G_- + 4|\Delta_1| G_+ ,
\end{align}
For general $p$, \new{ the steady state solutions obtained by setting $\dot{G}_\pm = 0$ are $G_+ = 0, -(3 - |\Delta_1|/J)/4, -(3+ \Delta_2/J )/4$.  Performing a linear stability analysis around these fixed points, the all identity solution $G_+ = G_- = 0$ only becomes an unstable fixed point for $|\Delta_1|/J \ge 3$, whereas it remains stable for weaker disorder. }    As a result, the phase transition in the mean-field theory occurs at the disorder strength $|\Delta_1| = 3 J$.  Beyond this value of disorder, the mean-field steady-state solution flows to an unphysical state; however, for weaker disorder, the physical mean-field solution remains stable.

\section{Statistical Mechanics Mapping Formalism}
\label{app:statmech}
Here, we review the statistical mechanics for the model with noise and antinoise, generalizing the mappings studied in Ref.~\cite{Nahum17,Bao20,Jian19,Dalzell2021,Dalzell2022}.
In this paper, we focus our attention to calculating second-moment quantities of a quantum state, which includes measures like fidelity, collision probability, and linear cross entropy. The circuit-averaged calculations of such quantities  lends itself to a statistical mechanical mapping to an Ising spin model. 

Consider a second moment measure $M$, averaged over circuits from ensemble $\mathcal{U}$. For a state a state $\rho_C$ of dimension $2^{2n} \times 2^{2n}$, generated using a circuit $C$ from an ensemble $\mathcal{U}$, a circuit-averaged second moment measure $M$ for can be written as a two-copy expectation of a linear operator $O_M$ of dimension $2^{4n} \times 2^{4n}$.
\begin{multline}
   \mathbb{E}_{C \in \mathcal{U}}[M[\rho_C]]  =  \mathbb{E}_{C \in \mathcal{U}}\left[\tr\left(O_M\rho_C \otimes \rho_C\right) \right] \\ = \tr\left(O_M \mathbb{E}_{C \in \mathcal{U}}\left[\rho_C \otimes \rho_C\right] \right)
\end{multline}
The circuit-averaged two-copy state $\mathbb{E}_{C \in \mathcal{U}}[\rho_C \otimes \rho_C]$, therefore, enables us to calculated second-moment measures.   

We consider circuit models that can be decomposed into a series of elementary two-qubit gates (noisy or noiseless), each drawn independently from the two-qubit Haar ensemble $\mathcal{U}_2$. The action of the circuit map on two-copies of an initial input state is, thus, given by
\begin{equation}
    \rho_C^{\otimes 2} = {C}_s \circ {C}_{s-1} \circ \cdots \circ {C}_{1}[\rho_0\otimes \rho_0],
\end{equation}
where the elementary single-qubit or two-qubit channels are indexed using integers $[1,s]$. For noiseless gates, ${C}_i[\rho_0\otimes \rho_0] = (C_i\otimes C_i) \rho (C_i^\dagger \otimes C_i^\dagger)$. We can model a noisy circuit by adding error channel $\mathcal{E}$ after each noiseless gate:
\begin{equation}
    \rho_C^{\otimes 2} = \mathcal{E}_s \circ {C}_{s-1} \circ \mathcal{E}_{s-1} {C}_{s-1} \circ \cdots \circ \mathcal{E}_1 \circ {C}_{1}[\rho_0\otimes \rho_0].
\end{equation}
In general, the error channel may act on any set of qubits. For our purposes, we assume that the error channel $\mathcal{E}_i$ acts on the qubit or the pair of qubits acted on by the noiseless gate $C_i$. Since we draw each gate from the random one-qubit or two-qubit Haar ensemble, we can replace the noiseless maps $C_i[\rho]$ with the gate-averaged map $\overline{C}_i[\sigma] = \mathbb{E}_{C \in \mathcal{U}_{1/2}} C[\sigma]$,
\begin{equation}
    \mathbb{E}_{C \in \mathcal{U}} [\rho_C \otimes \rho_C] = \mathcal{E}_s \circ \overline{C}_{s-1} \circ \mathcal{E}_{s-1} \overline{C}_{s-1} \circ \cdots \circ \mathcal{E}_1 \circ \overline{C}_{1}[\rho_0 \otimes \rho_0].
\end{equation}
The action of a random single-qubit gate $C$, on a state residing in the two-copy Hilbert space is given by 
\begin{equation}
    \overline{C}^{(1)}[\sigma] = \frac{\tr((1-S/2)\sigma)}{3} I + \frac{\tr((S-1/2)\sigma)}{3} S,
    \label{eq:single-q-gate}
\end{equation}
where $I$ and $S$ are the $4\times 4$ identity matrix and SWAP matrices, respectively. Similarly, the action of a random two-qubit gate on two copies of a qubit-pair is given by 
\begin{equation}
    \overline{C}^{(2)}[\sigma] = \frac{\tr((1-SS/4)\sigma)}{15} II + \frac{\tr((SS-1/4)\sigma)}{15} SS,
\end{equation}
where we use the shorthand $SS = S\otimes S$ and $II = I\otimes I$. 
If two copies of a quantum state can be represented by a string of $\rho^{\otimes 2} \in \{I, S\}^{n}$, a single-qubit gate acts on qubit $k$ by modifying the $j$th bit of the $I-S$ string using the transition rules
\begin{equation}
    I \to I \qquad S \to S,
\end{equation}
while leaving the rest of the bits in the string unchanged. Similarly, a two-qubit gate acting on qubits $j$ and $k$ modifies the $j$th and $k$th bits according to the transition rules:
\begin{equation}
 II \to II \qquad IS, SI \to \frac{2}{5}(II + SS)  \qquad {S}{S}. \to {S}{S}.
\end{equation}
A noiseless random circuit, therefore, can be represented as a linear operator acting on the reduced space spanned by basis elements in $\{I, S\}^{n}$, with each noiseless single-qubit gate given an identity map, and a two-qubit gate given by the transition matrix 
\begin{equation}
    T[\overline{C}^{(2)}] = \begin{pmatrix}1 & 2/5 & 2/5 & 0 \\ 0 & 0 & 0 & 0 \\ 0 & 0 & 0 & 0 \\ 0 & 2/5 & 2/5 & 1
    \end{pmatrix}.
\end{equation}

We can similarly, find the transition matrices corresponding to the noise and antinoise channel. The single-qubit depolarization channel with error rate $q$, given by the following map
\begin{equation}
    \mathcal{E}^{(1)}(\rho) = (1-q)\rho +  q\tr(\rho)\frac{\mathbb{1}}{2},
\end{equation}
acts on two copies of a qubit such that 
\begin{equation}
    \mathcal{E}[I] =  \mathcal{E}[I_2 \otimes I_2] =  \mathcal{E}^{(1)}[I_2] \otimes \mathcal{E}^{(1)}[I_2] = I_2 \otimes I_2 = I,
\end{equation}
where $I_2$ is a $2\times 2$ identity matrix. Similarly,
\begin{align*}
    \mathcal{E}&[S] \\
      &= (\mathcal{E}^{(1)}\otimes\mathcal{E})[I_2 \otimes I_2 + X\otimes X + Y \otimes Y + Z\otimes Z]/2\\
       &= \left[I_2 \otimes I_2 + (1-q)^2\left(X\otimes X + Y \otimes Y + Z\otimes Z\right)\right]/2\\
      &= \left[(1-(1-q)^2)/2 I + (1-q)^2S\right]
\end{align*}
where we have used the fact that $S = (I_2 \otimes I_2 + X\otimes X + Y \otimes Y + Z \otimes Z)/2$. The transition matrix corresponding to a depolarizing noise, in the statistical mechanical picture is given by, 
\begin{equation}
    T[\mathcal{E}_{q}] = \begin{pmatrix} 
       1 & (1-(1-q)^2)/2 \\
       0 & (1-q)^2 
    \end{pmatrix}
\end{equation}

Likewise, the antinoise channel of strength $q_a$, given by 
\begin{equation}
    \mathcal{A}^{(1)}(\rho) = \frac{1}{1-q_a}\left(\rho - q_a \tr(\rho)\frac{\mathbb{1}}{2} \right),
\end{equation}
acts on two-copies of a qubit such that 
\begin{align}
    \mathcal{A}[I] & =  (\mathcal{A}^{(1)}\otimes\mathcal{A}^{(1)})[I_2 \otimes I_2] =  I, \text{ and}  \\ 
      \mathcal{A}[S] &= \left[\left(1-\frac{1}{(1-q_a)^2}\right)\frac{I}{2} + \frac{1}{(1-q_a)^2}S\right],
\end{align}
giving a transition matrix 
\begin{equation}
   T[\mathcal{A}_{q_a}] = \begin{pmatrix} 
       1 & \left(1-(1-q_a)^{-2}\right)/2 \\
       0 & {(1-q_a)^{-2}}.
    \end{pmatrix}
\end{equation}

Concatenating the transition matrix for the noise channel and the antinoise channel gives the transition matrix for the composite channel. $T[\mathcal{A}_{q_a} \circ \mathcal{E}_q] = T[\mathcal{A}_{q_a}].T[\mathcal{E}_{q}]$. 

In our simulations, we start with an initial state drawn from the Haar ensemble or a random product state. A Haar random state is proportional to $I^{\otimes n} + S^{\otimes n}$ in the two-copy descrpition, while a random product state is proportional to $(I+S)^{\otimes n}$. We can then use the statistical mechanical formalism discussed above to evolve this state using the respective transition matrices for two-qubit gates, noise and antinoise channels.





\bibliography{main_prb.bbl}


\end{document}